\documentclass[11pt]{article}

\usepackage{amsmath, amssymb, amsthm}
\usepackage{hyperref}
\usepackage{geometry}
\usepackage{xcolor}
\usepackage[T1]{fontenc}

\newtheorem{theorem}{Theorem}
\newtheorem{lemma}{Lemma}
\newtheorem{corollary}{Corollary}
\newtheorem{definition}{Definition}

\newcommand{\Pcal}{\mathcal{P}}

\newcommand{\bp}{\mathbf{p}}
\newcommand{\bq}{\mathbf{q}}
\newcommand{\br}{\mathbf{r}}
\newcommand{\bw}{\mathbf{w}}
\newcommand{\bv}{\mathbf{v}}
\newcommand{\bx}{\mathbf{x}}
\newcommand{\by}{\mathbf{y}}
\newcommand{\cP}{{\cal P}}

\newcommand{\cL}{{\cal L}}

\newcommand{\remove}[1]{}

\usepackage{amsmath,amssymb,amsthm}

\title{The Sharma-Mittal Entropy is Subadditive and Supermodular 
on the Majorization Lattice}
\author{Roberto Bruno and Ugo Vaccaro\\
Department of Computer Science, University of Salerno\\
84084 Fisciano (SA), Italy\\
Email: {\tt \{rbruno,uvaccaro\}@unisa.it}}
\date{}

\begin{document}

\maketitle
\begin{abstract}
We prove that Sharma-Mittal entropy is a subadditive and supermodular function on the lattice of 
all $n$-dimensional probability distributions, ordered according to the partial order relation defined by majorization among vectors. Our result unifies and greatly extends analogous results 
presented in the literature for the Shannon
entropy, the Tsallis entropy, and the Rényi entropy.
    \end{abstract}

\section{Introduction}
The mathematical concept of majorization has a rich history, 
with applications across a wide range of disciplines. 
Majorization theory arose  in mathematical economics
\cite{AS}, where it was employed to
rigorously explain the vague notion that the components of a given vector are
``more nearly equal\rq\rq \  than the components of a different vector. Presently, majorization theory
finds applications in many areas, ranging from pure mathematics  to combinatorics
\cite{HA,Man,MO},
from information
and communication theory \cite{cv,CGV, A++, Mu, Sa1, Wi} to thermodynamics 
and quantum theory
\cite{bb,N, NV, S}, from mathematical chemistry \cite{bi} to optimization \cite{da}, among the others.

The quantification of uncertainty and diversity in complex systems relies heavily on generalized entropy measures. While the classical Shannon and Boltzmann-Gibbs entropies assume extensivity and independent subsystems, modern applications in non-extensive statistical mechanics and complex systems demand more flexible frameworks \cite{Gell}. The Sharma-Mittal entropy \cite{SM}, a  versatile two-parameter generalization, might serve this purpose. By decoupling the degree of non-extensivity from the deformation of the probability distribution, it recovers both the R\'enyi and Tsallis entropies as limiting cases \cite{masi}. Several other entropic functionals can be seen as particular cases of 
the Sharma-Mittal entropy \cite{Ta}.
Consequently, the Sharma-Mittal entropy has found applications across diverse domains, including holographic dark energy modeling in cosmology \cite{Sa}, topic modeling evaluation in natural language processing \cite{Ko}, and the algorithmic quantification of generative design variety \cite{A+}. Numerous  other applications of the Sharma-Mittal entropy in various fields of Physics
are discussed in \cite{Beck, Carann, Fi+,  Fra, Fra2,Saeed+++, Gha, Hu,Sa, Kaur,masi2,Maz,Mondaini,Nielsen,Jerin,Rud,scarf,Xuan}, and references quoted therein.

Simultaneously, the mathematical framework of majorization has proven useful for formally ordering probability 
distributions by their relative ``disorder'' or ``concentration.'' In particular, the probability simplex under 
the majorization pre-order forms a lattice---the \textit{majorization lattice}---equipped with well-defined meet
(greatest lower bound) and join (least upper bound) operations \cite{ba,cv}. The interplay between entropy 
measures and the majorization lattice has important physical implications, most notably in Quantum Theory. In the 
context of bipartite entanglement theory, state transformations under Local Operations and Classical Communication 
(LOCC), are strictly governed by majorization
\cite{N, NV}. Recent literature has extensively 
utilized the majorization lattice to model approximate bipartite entanglement transformations \cite{Bo1}, 
probabilistic pure state conversions \cite{De}, and the detection of high-dimensional quantum steering \cite{Ya}. Other applications of the majorization lattice
to Quantum Theory are discussed in the papers \cite{Bo2, Bo3,  LiQi, Ma++}.

Despite several studies, the structural properties of generalized entropies on the majorization lattice remain incompletely charted. Cicalese and Vaccaro \cite{cv} established that the Shannon entropy is 
subadditive and supermodular  on the majorization lattice. Harremo{\"e}s \cite{Har2} greatly simplified 
the proofs of the results in \cite{cv}. Recently, Yadav and Shkel \cite{yadav2025} extended 
the results of \cite{cv} by proving the subadditivity and supermodularity of the  R\'enyi entropy, and Bhatia \textit{et al.}\ \cite{ba} proved 
the subadditivity of the Tsallis entropy. 


In this paper, we  characterize the subadditive, superadditive and supermodular behavior of the Sharma-Mittal entropy $S_{\alpha,\beta}(\mathbf{p})$ on the majorization lattice. Specifically, 
we prove that the Sharma-Mittal entropy $S_{\alpha,\beta}(\mathbf{p})$ is subadditive for $\alpha\geq 0$ and $\beta\geq 1$, 
is superadditive for $\alpha<0$ and $\beta\leq 1$, and  is supermodular for $\alpha> 0$ and $\beta\leq \alpha$.
This provides  a unified treatment and considerable extension of the results in \cite{ba,cv, yadav2025},
since the Sharma-Mittal entropy encompasses all the information measures studied therein.
We also  investigate the regime where $\beta > \alpha$. We demonstrate a fundamental structural breakdown: by constructing explicit counterexamples (specifically evaluated at $\alpha=2$, $\beta=3$), we prove that the Sharma-Mittal entropy generally loses both supermodularity and submodularity when $\beta > \alpha$ and $\alpha>0$.

\section{Preliminaries}
Let 
$\cP_n = \{\bp=(p_1,\dots,p_n): p_i\geq  0,  \sum_{i=1}^n p_i = 1\}$
be the $(n-1)$-dimensional probability simplex.
Throughout this paper, we assume that all  vectors $\bp=(p_1,\dots,p_n)\in \cP_n$ are ordered in non increasing fashion, that is,
$p_1\geq \dots\geq p_n$.
We recall the basics of majorization theory \cite{MO}.

\begin{definition} 
    For any $\bp,\bq \in \cP_n$,  we say that $\bp$ is \textit{majorized} by $\bq$ (equivalently, 
that $\bq$ majorizes $\bp$), denoted by $\mathbf{p} \preceq \mathbf{q}$, if it holds that
\begin{equation}\label{p<q}
         \sum_{i=1}^k p_i \leq \sum_{i=1}^k q_i,  \quad \mbox{ for } \ k=1,\ldots,n.
         \end{equation}
\end{definition}
In the more general case, in which
$\bp$ and $\bq$ have a different number of components, we pad the shorter vectors with an appropriate number 
of 0's and apply the same definition.
The majorization relation $\preceq$ is a partial order relation on $\cP_n, $ that is, $(\cP_n,\preceq)$ is a poset.

\begin{definition}\label{def:schurconvex}
A function $\phi:\cP_{n}\rightarrow \mathbb{R}$ is said to be Schur-convex if for every $\bp,\bq$ $\in$ $\cP_{n}$ satisfying $\bp \preceq \bq$, it holds that 
\begin{align}
 \phi(\bp)\leq \phi(\bq)
\end{align}
The function $\phi$ is  {Schur-concave} if $-\phi$ is Schur-convex.
\end{definition}
We  recall \cite{DP} that a lattice is a quadruple
$⟨\cL,\sqsubseteq,\wedge,\vee⟩$  where $\cL$ is a set, $\sqsubseteq$ is a partial ordering on $\cL$, and for all $a,b\in \cL$ there is a unique greatest lower bound (glb) $a\wedge b$ and a unique least upper bound (lub) 
$a\vee b$. More precisely, $a\wedge b$ and $a\vee b$ are the elements of $\cL$ that satisfy the conditions
$$a\wedge b\sqsubseteq a, a\wedge b\sqsubseteq b, \quad  a\sqsubseteq a\vee b, b\sqsubseteq a\vee b,$$
and for each $c$ and $d$ such that 
$c\sqsubseteq a, c\sqsubseteq b, a\sqsubseteq d, b\sqsubseteq d$ one has that
$$c\sqsubseteq a \wedge b \quad {\hbox{and}}\quad a\vee b\sqsubseteq d.$$

Bapat~\cite{bapat} showed that the partially ordered set  $(\mathcal P_n,\preceq)$ induced by majorization is a  lattice. Cicalese and Vaccaro~\cite{cv} gave explicit constructions for the greatest lower bound (glb) and the least upper bound (lub) of any two elements. We recall their algorithm below.
Given $\bp$ and $\bq$ in $\mathcal{P}_n$, the glb $\bp \wedge \bq=\br =(r_1, \ldots , r_n) $ is given by
\begin{align}
    \sum^{k}_{i=1} r_i = \min\left\{\sum_{i=1}^{k}p_i,\sum_{i=1}^{k}q_i \right\}.\label{eq:glb}
\end{align}
Define a vector $\beta(\bp,\bq)=\bw \in \mathbb{R}^{n}$ such that
\begin{align}
   \sum^{k}_{i=1} w_i=\max\left\{\sum_{i=1}^{k}p_i,\sum_{i=1}^{k}q_i \right\}. \label{eq:lub}
\end{align}
If $\bw \in\mathcal P_n$, then the lub $\bp\vee \bq$ is equal to $\bw$. Otherwise, $\bp \vee \bq$ is obtained by repeatedly replacing each maximal consecutive block of coordinates of $\bw$ that violates the non-increasing order by its average, until the resulting vector lies in $\mathcal P_n$. 
Equivalent methods to obtain $\bp \wedge \bq$ and $\bp \vee \bq$  in terms of Lorents curves are presented in \cite{ cuff, Har2, yadav2025}.

\subsection{Entropies}
Le $\bp \in \cP_n$. We assume that all logarithms in this paper are of base 2.
The Shannon entropy is defined as 
$$H(\bp)=-\sum_{i=1}^np_i\log p_i.$$
The Rényi entropy  \cite{Re} of order $\alpha \in [0,\infty]$,
is defined as
$$
 H_{\alpha}(\mathbf{p})= \frac{1}{1-\alpha} \log\left( \sum_{i=1}^{n} p_i^ {\alpha}\right).
$$
For $\alpha \rightarrow 1$, the Rényi entropy  reduces to the Shannon entropy. The Rényi entropy is Schur‑concave.

For $\alpha\in [0,\infty)$,  the Tsallis entropy \cite{Ts} is defined as
$$T_\alpha(\bp)=\frac{1-\sum_{i=1}^np_i^\alpha}{\alpha-1}. $$
For $\alpha \rightarrow 1$, also the Tsallis  entropy   reduces to the Shannon entropy. 

For $\alpha, \beta\in [0,\infty)$, the Sharma-Mittal  entropy \cite{masi,SM} is defined as 
\begin{equation}\label{eq:Sharma-Mittal}
    S_{\alpha,\beta}(\bp)=\frac{1}{1-\beta}\left[\left 
(\sum_{i=1}^np_i^\alpha\right )^{\frac{1-\beta}{1-\alpha}}-1\right ].
\end{equation}
One can see that $\lim_{\beta\to 1}S_{\alpha,\beta}(\bp)=\ln(2)H_\alpha(\mathbf{p})$
and for $\alpha\neq 1$, it holds that  $\lim_{\beta\to \alpha}S_{\alpha,\beta}(\bp)=T_\alpha(\mathbf{p})$.
From the definition of $H_\alpha(\mathbf{p})$, one can see 
that $\sum_{i=1}^n p_i^\alpha=2^{(1-\alpha)H_\alpha(\bp)}$. Therefore,
one can rewrite the Sharma-Mittal entropy as follows 
\begin{equation}\label{eq:sharma-mittal-renyi}
    S_{\alpha,\beta}(\bp) = \frac{1}{1-\beta} 
    \left[ \left( 2^{(1-\alpha)H_\alpha(\bp)} \right)^{\frac{1-\beta}{1-\alpha}} - 1 \right] = \frac{2^{(1-\beta)H_\alpha(\bp)} - 1}{1-\beta}.
\end{equation}
Equivalently, if one defines the function $\phi_\beta(x)$
as 
\begin{equation}\label{phib}
\phi_\beta(x)=\begin{cases}
\frac{2^{(1-\beta)x}-1}{1-\beta}, & \mbox{ for }\beta\neq 1,\\
\ln(2)x, &\mbox{ for } \beta=1, 
\end{cases}
\end{equation}
then one has
\begin{equation}\label{S=R}
S_{\alpha,\beta}(\bp) =\phi_\beta(H_\alpha(\bp)).
\end{equation}
\section{Subadditivity of the Sharma-Mittal entropy on the majorization lattice}

We first recall the notion of subadditive (resp. superadditive) functions on an abstract lattice.
\begin{definition}
    A real-valued function $\phi:\mathcal{P} \rightarrow \mathbb{R}$ defined on a lattice $(\mathcal{P},\preceq,\wedge,\vee)$ is subadditive (resp. superadditive) if $\forall$ $\bx$, $\by$ $\in$ $\mathcal{P}$, 
    it holds that
    \begin{align}
        \phi(\bx \wedge \by) \leq(\text{resp. } \geq) \ \phi(\bx)+\phi(\by).
    \end{align}
    \end{definition}
    The authors of \cite{cv} proved that the Shannon entropy is subadditive on the majorization lattice
   $(\cP_n,\preceq,\wedge,\vee)$.  This result was later extended to the Tsallis entropy in \cite{ba}. 
    Recently, the authors of \cite{yadav2025} proved the analogous result for the  R{\'e}nyi entropy.

In this section, we extend the above results to the more general case of Sharma-Mittal entropies 
$S_{\alpha,\beta}$. To this end, we first establish some preliminary results. In the following, whenever $\alpha < 0$, we assume that $p_i > 0$ for all $i$, since otherwise the Sharma-Mittal entropy $S_{\alpha,\beta}(\bp)$ would not be well-defined.

\begin{lemma}\label{lemma:schur_convexity}
    The Sharma-Mittal entropy 
    $$S_{\alpha,\beta}(\bp)=\frac{1}{1-\beta}\left[\left 
(\sum_{i=1}^np_i^\alpha\right )^{\frac{1-\beta}{1-\alpha}}-1\right ]$$
    is Schur-concave for $\alpha\geq 0$ and Schur-convex for $\alpha<0$.
\end{lemma}
\begin{proof}
To determine the Schur-concavity or Schur-convexity of the Sharma-Mittal entropy $S_{\alpha,\beta}(\bp)$, we apply the Schur-Ostrowski criterion \cite[Ch. 3, Thm. A.4]{MO}: a function $\phi:D\subset\mathbb{R}^n\to \mathbb{R}$ is Schur-concave (resp. Schur-convex) if and only if $\phi$ is symmetric (that is, it is invariant under any permutation of its arguments), continuously differentiable, and for all $i\neq j$ and $z=(z_1,\dots,z_n)\in D$ it holds that 
\begin{equation*}
    (z_i - z_j)\left( \frac{\partial \phi}{\partial p_i} - \frac{\partial \phi}{\partial p_j} \right) \leq 0 \quad (\text{resp. } \geq 0).
\end{equation*}
Thus, because $S_{\alpha,\beta}(\bp)$ is symmetric and continuously differentiable, it follows that it is Schur-concave (resp. Schur-convex) if and only if for all $i \neq j$:
\begin{equation}\label{eq:nec_suff_cond}
    (p_i - p_j)\left( \frac{\partial S_{\alpha,\beta}}{\partial p_i} - \frac{\partial S_{\alpha,\beta}}{\partial p_j} \right) \leq 0 \quad (\text{resp. } \geq 0).
\end{equation}
To this end, let us compute the partial derivative of $S_{\alpha,\beta}(\mathbf{p})$ with respect to a generic component $p_i$ of $\bp$. For the sake of notation, let $A=\sum_{k=1}^n p_k^\alpha$. It follows:
$$\frac{\partial S_{\alpha,\beta}}{\partial p_i} = \frac{1}{1-\beta} \left[ \frac{1-\beta}{1-\alpha} A^{\frac{1-\beta}{1-\alpha} - 1} \cdot \alpha p_i^{\alpha-1} \right] = \frac{\alpha}{1-\alpha} A^{\frac{\alpha-\beta}{1-\alpha}} p_i^{\alpha-1}.$$
Thus, we need to study the behaviour of the following expression:
\begin{equation}\label{eq:diff}
    (p_i-p_j)\left(\frac{\partial S_{\alpha,\beta}}{\partial p_i} - \frac{\partial S_{\alpha,\beta}}{\partial p_j}\right) =(p_i-p_j) \frac{\alpha}{1-\alpha} A^{\frac{\alpha-\beta}{1-\alpha}} \left( p_i^{\alpha-1} - p_j^{\alpha-1} \right).
\end{equation}
Since the term $A^{\frac{\alpha-\beta}{1-\alpha}}$ is always strictly positive, it does not affect the sign of the expression \eqref{eq:diff}. Thus, we just need to study the sign of the function
$$
    g(\alpha)=(p_i-p_j) \frac{\alpha}{1-\alpha}  \left( p_i^{\alpha-1} - p_j^{\alpha-1} \right).
$$
We now examine how $g(\alpha)$ behaves with respect to $\alpha$. We assume $p_i\neq p_j$ to avoid trivialities.

\noindent
\textit{Case:} $\alpha\geq 0$. We consider two subcases: $0<\alpha<1$ and $\alpha>1$. For the boundary cases $\alpha=0$ and $\alpha=1$, we observe that $S_{\alpha,\beta}$ is Schur-concave. In fact, for $\alpha=0$, it holds trivially since $g(0)=0$, while for $\alpha=1$, it holds because $\lim_{\alpha\to 1} g(\alpha)=-(p_i-p_j)(\ln p_i-\ln p_j)\leq0$.

Let $0<\alpha<1$. In this case, since the function $x\to x^{\alpha-1}$ is strictly decreasing, the factors
\begin{equation*}
    (p_i-p_j)\quad\text{and}\quad (p_i^{\alpha-1} - p_j^{\alpha-1})
\end{equation*}
always have opposite signs. 
Hence,
\begin{equation*}
    (p_i-p_j)(p_i^{\alpha-1} - p_j^{\alpha-1})< 0.
\end{equation*}
Therefore, since the term $\frac{\alpha}{1-\alpha}> 0$, it follows that $g(\alpha)< 0$.

Let $\alpha>1$. Since the function $x\to x^{\alpha-1}$ is strictly increasing,
the factors
\begin{equation*}
    (p_i-p_j)\quad\text{and}\quad (p_i^{\alpha-1} - p_j^{\alpha-1})
\end{equation*}
always have the same sign. 
Hence,
\begin{equation*}
    (p_i-p_j)(p_i^{\alpha-1} - p_j^{\alpha-1})> 0.
\end{equation*}
From this, since the term $\frac{\alpha}{1-\alpha}<0$, it follows that $g(\alpha)< 0$.

\noindent
\textit{Case}: $\alpha<0$. Since the function $x\to x^{\alpha-1}$ is strictly decreasing, the factors
\begin{equation*}
    (p_i-p_j)\quad\text{and}\quad (p_i^{\alpha-1} - p_j^{\alpha-1})
\end{equation*}
always have opposite signs. 
Thus, since
\begin{equation*}
    (p_i-p_j)(p_i^{\alpha-1} - p_j^{\alpha-1})< 0,
\end{equation*}
and the term $\frac{\alpha}{1-\alpha}< 0$, it follows that $g(\alpha)>0$. 
\end{proof}

The following technical Lemma is needed in the proof of Theorem \ref{thm:sub_superadd}, that represents the main result
of this section.
\begin{lemma}\label{lemma:almost_additivity}
Let $\bp,\bq \in \cP_n$, and let $\bp\otimes\bq$ denote the tensor product of $\bp$ and $\bq$, that is, 
$$
    (\bp\otimes\bq)_{i,j} = p_iq_j \quad\forall i,j=1,\dots,n.
$$
Then, for any $\alpha,\beta\in\mathbb{R}$, it holds that
\begin{equation}
    S_{\alpha,\beta}(\bp\otimes\bq)=S_{\alpha,\beta}(\bp)+S_{\alpha,\beta}(\bq)+(1-\beta)S_{\alpha,\beta}(\bp)S_{\alpha,\beta}(\bq).
\end{equation}
\end{lemma}
\begin{proof}
    We first observe that as  $\beta\to 1$ and as $(\alpha,\beta)\to(1,1)$, the Sharma-Mittal entropy reduces to the R\'enyi and Shannon entropies, respectively, for which the equality holds. Thus, we can restrict our attention to $\alpha,\beta\in\mathbb{R}\setminus\{1\}$. 
    We observe that from \eqref{eq:Sharma-Mittal}, we have that
    $$1+(1-\beta)S_{\alpha,\beta}(\bp)=\left(\sum_{i=1}^n p_i^\alpha\right)^{\frac{1-\beta}{1-\alpha}}.$$
    Thus, it follows that
    \begin{align}
        1+(1-\beta)S_{\alpha,\beta}(\bp\otimes\bq)&=\left(\sum_{i,j} (p_iq_j)^\alpha\right)^{\frac{1-\beta}{1-\alpha}}\nonumber\\
        &=\left(\sum_{i,j} p_i^\alpha q_j^\alpha\right)^{\frac{1-\beta}{1-\alpha}}\nonumber\\
        &=\left(\left(\sum_{i} p_i^\alpha\right) \left(\sum_{j}q_j^\alpha\right)\right)^{\frac{1-\beta}{1-\alpha}}\nonumber\\
        &=\left(\sum_{i} p_i^\alpha\right)^{\frac{1-\beta}{1-\alpha}}\left(\sum_{j}q_j^\alpha\right)^{\frac{1-\beta}{1-\alpha}}\nonumber\\
        &=\left(1+(1-\beta)S_{\alpha,\beta}(\bp)\right)\left(1+(1-\beta)S_{\alpha,\beta}(\bq)\right)\nonumber\\
        &=1+(1-\beta)S_{\alpha,\beta}(\bp)+(1-\beta)S_{\alpha,\beta}(\bq)+(1-\beta)^2S_{\alpha,\beta}(\bp)S_{\alpha,\beta}(\bq).\label{eq:last_step}
    \end{align}
    Rearranging \eqref{eq:last_step} to isolate $S_{\alpha,\beta}(\bp\otimes\bq)$ and dividing by $(1-\beta)$, we obtain
    \begin{equation*}
        S_{\alpha,\beta}(\bp\otimes\bq)=S_{\alpha,\beta}(\bp)+S_{\alpha,\beta}(\bq)+(1-\beta)S_{\alpha,\beta}(\bp)S_{\alpha,\beta}(\bq),
    \end{equation*}
    which concludes the proof.
\end{proof}

We now present the main result of the section.
\begin{theorem}\label{thm:sub_superadd}
Let $\beta\in \mathbb{R}$, and let $\bp,\bq\in \cP_n$. Then, for $\alpha\geq 0$, it holds that
\begin{equation}
    S_{\alpha,\beta}(\bp\wedge\bq) \leq S_{\alpha,\beta}(\bp)+S_{\alpha,\beta}(\bq)+(1-\beta)S_{\alpha,\beta}(\bp)S_{\alpha,\beta}(\bq),
\end{equation}
and for $\alpha<0$, it holds that
\begin{equation}
    S_{\alpha,\beta}(\bp\wedge\bq) \geq S_{\alpha,\beta}(\bp)+S_{\alpha,\beta}(\bq)+(1-\beta)S_{\alpha,\beta}(\bp)S_{\alpha,\beta}(\bq).
\end{equation}
\remove{
\begin{equation}
    S_{\alpha,\beta}(\bp\wedge\bq) 
    \begin{cases}
        \leq S_{\alpha,\beta}(\bp)+S_{\alpha,\beta}(\bq)+(1-\beta)S_{\alpha,\beta}(\bp)S_{\alpha,\beta}(\bq) & \text{if } \alpha \geq 0, \\[1ex]
        \geq S_{\alpha,\beta}(\bp)+S_{\alpha,\beta}(\bq)+(1-\beta)S_{\alpha,\beta}(\bp)S_{\alpha,\beta}(\bq) & \text{if } \alpha < 0.
    \end{cases}
\end{equation}
}
\end{theorem}

\begin{proof}
    First, we observe that $\bp$ and $\bq$ are aggregations of the tensor product $\bp\otimes\bq$, that is, $p_i=\sum_{j} p_iq_j$ for all $i$ and $q_j=\sum_i p_iq_j$ for all $j$. In the context of majorization, an aggregation of a distribution always majorizes the original one. Consequently, we have
    \begin{equation*}
        \bp\otimes\bq\preceq \bp\quad \text{and} \quad\bp\otimes\bq\preceq\bq.
    \end{equation*}
    Moreover, by definition of the greatest lower bound $\bp\wedge\bq$, it immediately follows that 
    \begin{equation}\label{eq:prod_maj}
        \bp\otimes\bq\preceq\bp \wedge\bq.
    \end{equation}
    Now, from \eqref{eq:prod_maj} and Lemma \ref{lemma:schur_convexity}, we have that for $\alpha\geq 0$, the Sharma-Mittal entropy $S_{\alpha,\beta}$ is Schur-concave, yielding
    \begin{equation}\label{eq:conc}
        S_{\alpha,\beta}(\bp\wedge\bq)\leq  S_{\alpha,\beta}(\bp\otimes\bq).
    \end{equation}
    Conversely, for $\alpha<0$, it is Schur-convex, yielding
    \begin{equation}\label{eq:convex}
        S_{\alpha,\beta}(\bp\wedge\bq)\geq  S_{\alpha,\beta}(\bp\otimes\bq).
    \end{equation}
    Finally, applying Lemma \ref{lemma:almost_additivity} to expand $S_{\alpha,\beta}(\bp\otimes\bq)$ in \eqref{eq:conc} and \eqref{eq:convex} yields the desired inequalities, concluding the proof.
\end{proof}

As a direct consequence of Theorem \ref{thm:sub_superadd}, we obtain the following subadditivity and superadditivity properties for the Sharma-Mittal entropy, depending on the values of the parameters $\alpha$ and $\beta$.
\begin{corollary}\label{cor:sub}
    For every $\alpha\geq 0$ and $\beta\geq 1$, the Sharma-Mittal entropy $S_{\alpha,\beta}$ is subadditive on the majorization lattice, that is, for any $\bp,\bq\in\cP_n$, it holds that
    \begin{equation}\label{eq:rsuba}
     S_{\alpha,\beta}(\mathbf{p}\wedge\mathbf{q}) \le S_{\alpha,\beta}(\mathbf{p})+S_{\alpha,\beta}(\mathbf{q}).  
    \end{equation}
\end{corollary}
\begin{proof}
    It follows from Theorem \ref{thm:sub_superadd} since the Sharma-Mittal entropy is non-negative and for $\beta\geq1$ the term $(1-\beta)S_{\alpha,\beta}(\bp)S_{\alpha,\beta}(\bq)$ is smaller or equal to $0$.
\end{proof}

For $\beta\to 1$, we recover Theorem 1 of \cite{yadav2025}, for $\beta\to \alpha$ we obtain Theorem 3.3  of \cite{ba}, and 
for $(\alpha,\beta)\to (1,1)$, we get Theorem 2 of \cite{cv}.

Although the classical definition of the Sharma-Mittal entropy posits that $\alpha,\beta\geq 0$, we state the following result for
the record.
\begin{corollary}\label{cor:super}
    For every $\alpha< 0$ and $\beta\leq 1$, the Sharma-Mittal entropy $S_{\alpha,\beta}$ is superadditive on the majorization lattice, that is, for any $\bp,\bq\in\cP_n$, it holds that
    \begin{equation}\label{eq:rsupera}
     S_{\alpha,\beta}(\mathbf{p}\wedge\mathbf{q}) \geq S_{\alpha,\beta}(\mathbf{p})+S_{\alpha,\beta}(\mathbf{q}).  
    \end{equation}
\end{corollary}
\begin{proof}
    It follows from Theorem \ref{thm:sub_superadd} since the Sharma-Mittal entropy is non-negative and for $\beta\leq1$ the term $(1-\beta)S_{\alpha,\beta}(\bp)S_{\alpha,\beta}(\bq)$ is greater or equal to $0$.
\end{proof}

\section{Supermodularity  of the Sharma-Mittal entropy on the majorization lattice}
We  recall the definition of  supermodular  functions on an abstract lattice.
\begin{definition}
    A real-valued function $\phi:\mathcal{P} \rightarrow \mathbb{R}$ defined on a lattice $(\mathcal{P},\preceq,\wedge,\vee)$ is supermodular  if $\forall$ $\bx$, $\by$ $\in$ $\mathcal{P}$, 
    it holds that
    \begin{align}
        \phi(\bx)+\phi(\by) \leq \phi(\bx \wedge \by)+\phi(\bx\vee \by).
    \end{align}
    \end{definition}

    The paper  \cite{cv} proved supermodularity for the Shannon entropy,  \cite{yadav2025} proved supermodularity 
    for Rényi and Tsallis entropy.
In this section, we extend these results to the more general case of Sharma-Mittal entropies.

\begin{theorem}\label{thm:sper}
Let $\alpha >0 $ and $\beta\le \alpha$. Then the Sharma--Mittal entropy
    is supermodular on the majorization lattice, that is, for any $\mathbf{p},\mathbf{q}\in \cP_n$, it holds that
\[
S_{\alpha,\beta}(\bp)+S_{\alpha,\beta}(\bq)
\le
S_{\alpha,\beta}(\bp\wedge \bq)+S_{\alpha,\beta}(\bp\vee \bq).
\]
    
\end{theorem}
\begin{proof}
    Let
\[
\bw=\bp\wedge \bq,
\qquad
\bv=\bp\vee \bq,
\]
and define the function $g_\alpha(\bp):\cP_n\to \mathbb{R}$ as
\begin{equation}
    g_\alpha(\bp)=\sum_{i=1}^n p_i^\alpha.
\end{equation}
Recalling the definition of the Tsallis entropy of order $\alpha\in(0,\infty)$
$$
T_\alpha(\bp)=\frac{1-\sum_{i=1}^np_i^\alpha}{\alpha-1},$$
it follows that 
\begin{equation}
    g_\alpha(\bp)=(1-\alpha)T_\alpha(\bp)+1.
\end{equation}
Moreover, define the function $h_{\alpha,\beta}:\mathbb{R^+}\to \mathbb{R}$ as
\begin{equation}\label{eq:h_def}
    h_{\alpha,\beta}(x)=\frac{1}{1-\beta}\left(x^{\frac{1-\beta}{1-\alpha}}-1\right).
\end{equation}
We can express the Sharma-Mittal entropy as a function of $g_\alpha(\bp)$ as follows:
\begin{equation}\label{eq:sharma_m_as_tsa}
    S_{\alpha,\beta}(\bp)=h_{\alpha,\beta}(g_\alpha(\bp)).
\end{equation}

We now proceed by considering three distinct cases depending on the value of $\alpha$.

\noindent
\textbf{Case:} $0<\alpha<1$. In this case, since $1-\alpha> 0$, and since the Tsallis entropy is Schur-concave and subadditive on the majorization lattice \cite{ba}, and supermodular for $\alpha\in[0,\infty)$ \cite{yadav2025}, we have that  the function $g_\alpha(\bp)$ is also Schur-concave, subadditive and supermodular. 

Moreover, since $\alpha\in(0,1)$ and $\beta\leq \alpha$, the exponent $\frac{1-\beta}{1-\alpha}$
in (\ref{eq:h_def}) is greater or equal to $1$, and both the first derivative $h'_{\alpha,\beta}(x)=\frac{1}{1-\alpha}x^{\frac{1-\beta}{1-\alpha}-1}$ and the second derivative $h''_{\alpha,\beta}(x)=\frac{\frac{1-\beta}{1-\alpha}-1}{1-\alpha}x^{\frac{1-\beta}{1-\alpha}-2}$ are non negative. Thus, both the function $h_{\alpha,\beta}(x)$ and its derivative $h'_{\alpha,\beta}(x)$ are increasing in $x$.

\smallskip
Let 
\[
a=g_\alpha(\bp),
\qquad
b=g_\alpha(\bq),
\qquad
c=g_\alpha(\bw),
\qquad
d=g_\alpha(\bv).
\]
Because $g_\alpha$ is Schur-concave and
\[
\bw\preceq \bp,\; \bq\preceq \bv,
\]
it follows
\[
c\ge a,
\qquad
c\ge b,
\qquad
a\ge d,
\qquad
b\ge d.
\]
Moreover, without loss of generality, assuming $a\geq b$, yields
\begin{equation}\label{eq:order}
c\ge a\ge b\ge d.
\end{equation}
Now, since $g_\alpha$ is supermodular on the majorization lattice, we have that
\begin{equation}\label{eq:g_supermod}
    a+b\leq c+d.
\end{equation}
Let us define 
\[
\delta=(c-a)\ge0.
\]
From (\ref{eq:g_supermod}), it holds that
\[
d-b\ge -\delta.
\]
Hence, there exists $\varepsilon\ge0$ such that
\begin{equation}\label{eq:d}
d=b-\delta+\varepsilon.
\end{equation}
Thus, we have that $c=a+\delta$ and $d=b-\delta+\varepsilon$.

Define the function $G(t)$ as
\[
G(t)=h_{\alpha,\beta}(a+t)+h_{\alpha,\beta}(b-t),
\qquad
0\le t\le \delta,
\]
where $h_{\alpha,\beta}$ is defined in (\ref{eq:h_def}). First, we observe that the argument $b-t$ is always non negative for $t\in[0,\delta]$. This holds because $g_\alpha$ is subadditive and, therefore, since $g_\alpha(\bp\wedge\bq)=c\leq a+b=g_\alpha(\bp)+g_\alpha(\bq)$, we have $b-t\geq b-\delta=a+b-c\geq0$ for $t\in[0,\delta]$.

Let 
$$
    G'(t)=h'_{\alpha,\beta}(a+t)-h'_{\alpha,\beta}(b-t)
$$
be the derivative of $G(t)$.
Since, from the chain of inequalities in (\ref{eq:order}), we have that
\[
a+t\ge b-t,
\]
and since $h'_{\alpha,\beta}$ is increasing, it follows that
\[
G'(t)\ge0.
\]
Thus, the function $G$ is increasing and, evaluating $G(t)$ at $t=\delta$ and $t=0$ we obtain

\begin{equation}\label{eq:G}
    G(\delta)=h_{\alpha,\beta}(a+\delta)+h_{\alpha,\beta}(b-\delta)\geq h_{\alpha,\beta}(a)+h_{\alpha,\beta}(b)=G(0).
\end{equation}
Moreover, since $h_{\alpha,\beta}$ is increasing and $\varepsilon\geq 0$, we have that
\begin{equation}\label{eq:b+eps}
    h_{\alpha,\beta}(b-\delta+\varepsilon)\geq h_{\alpha,\beta}(b-\delta).
\end{equation}
Combining \eqref{eq:b+eps} with \eqref{eq:G}, we finally get
\begin{equation*}
    h_{\alpha,\beta}(c)+h_{\alpha,\beta}(d)=h_{\alpha,\beta}(a+\delta)+h_{\alpha,\beta}(b-\delta+\varepsilon)\geq h_{\alpha,\beta}(a)+h_{\alpha,\beta}(b).
\end{equation*}
Recalling the definitions of $a,b,c$ and $d$, and \eqref{eq:sharma_m_as_tsa} we obtain
\[
S_{\alpha,\beta}(\bw)+S_{\alpha,\beta}(\bv)
\ge
S_{\alpha,\beta}(\bp)+S_{\alpha,\beta}(\bq),
\]
which demonstrates that  $S_{\alpha,\beta}$ is supermodular in the range $0<\alpha < 1$ and $\beta\le \alpha$.

\smallskip
\noindent
\textbf{Case:} $\alpha>1$.  In this case, since $1-\alpha< 0$, we have that the function $g_\alpha(\bp)$ is Schur-convex and submodular on the majorization lattice. Moreover, since $\alpha>1$ and $\beta\leq \alpha$, we have that the first derivative $h'_{\alpha,\beta}(x)=\frac{1}{1-\alpha}x^{\frac{1-\beta}{1-\alpha}-1}$ is negative and the second derivative $h''_{\alpha,\beta}(x)=\frac{\frac{1-\beta}{1-\alpha}-1}{1-\alpha}x^{\frac{1-\beta}{1-\alpha}-2}$ is positive. Thus, the function $h(x)$ is convex and decreasing in $x$, and its derivative $h'(x)$ is increasing in $x$.
Let 
\[
a=g_\alpha(\bp),
\qquad
b=g_\alpha(\bq),
\qquad
c=g_\alpha(\bw),
\qquad
d=g_\alpha(\bv).
\]
Because $g_\alpha$ is Schur-convex and
\[
\bw\preceq \bp,\; \bq\preceq \bv,
\]
it follows
\[
c\leq a,
\qquad
c\leq b,
\qquad
a\leq d,
\qquad
b\leq d.
\]
Moreover, without loss of generality, assuming $a\geq b$, yields
\begin{equation}\label{eq:order_convex}
c\leq b\leq a\leq d.
\end{equation}
Since $g_\alpha$ is submodular on the majorization lattice, we have that
\begin{equation}\label{eq:g_submodular}
    a+b\geq c+d.
\end{equation}
Let us define 
\[
\delta=(a-c)\ge0.
\]
From (\ref{eq:g_submodular}), it holds that
\[
d-b\leq \delta.
\]
Hence, there exists $\varepsilon\ge0$ such that
\begin{equation}\label{eq:d_convex}
d=b+\delta-\varepsilon.
\end{equation}
Thus, we have that $c=a-\delta$ and $d=b+\delta-\varepsilon$.
Define the function $G(t)$ as
\[
G(t)=h_{\alpha,\beta}(a-t)+h_{\alpha,\beta}(b+t),
\qquad
0\le t\le \delta,
\]
where $h_{\alpha,\beta}$ is defined in (\ref{eq:h_def}). First, we observe that the argument $a-t\geq a-\delta=c$ is always non negative for $t\in[0,\delta]$.
Let 
$$
    G'(t)=-h'_{\alpha,\beta}(a-t)+h'_{\alpha,\beta}(b+t)
$$
be the derivative of $G(t)$. We observe that $G'(t)\geq 0$ for all $t$ which satisfy
$$
    h'_{\alpha,\beta}(b+t)\geq h'_{\alpha,\beta}(a-t).
$$
But since $h'$ is increasing, this implies that $G'(t)\geq 0$ for all $t$ such that
$$
    t\geq \frac{a-b}{2}.
$$
Thus, the function $G$ is increasing for $t\in[\frac{a-b}{2},\delta]$. Moreover, evaluating $G(t)$ at $t=a-b\leq a-c=\delta$, we have that
\begin{equation}\label{eq:g(a-b)}
    G(a-b)=h_{\alpha,\beta}(a-(a-b))+h_{\alpha,\beta}(b+(a-b))=h_{\alpha,\beta}(b)+h_{\alpha,\beta}(a)=G(0).
\end{equation}
Therefore, from \eqref{eq:g(a-b)}, we have that $G(0)=G(a-b)$, and since the function $G$ is increasing for $t\in[\frac{a-b}{2},\delta]$, it follows that
\begin{equation}\label{eq:G_convex}
    G(\delta)=h_{\alpha,\beta}(a-\delta)+h_{\alpha,\beta}(b+\delta)\geq h_{\alpha,\beta}(a)+h_{\alpha,\beta}(b)=G(0).
\end{equation}
Moreover, since $h_{\alpha,\beta}$ is decreasing and $\varepsilon\geq 0$, we have that
\begin{equation}\label{eq:b-eps}
    h_{\alpha,\beta}(b+\delta-\varepsilon)\geq h_{\alpha,\beta}(b+\delta).
\end{equation}
Combining \eqref{eq:b-eps} with \eqref{eq:G_convex}, we finally get
\begin{equation*}
    h_{\alpha,\beta}(c)+h_{\alpha,\beta}(d)=h_{\alpha,\beta}(a-\delta)+h_{\alpha,\beta}(b+\delta-\varepsilon)\geq h_{\alpha,\beta}(a)+h_{\alpha,\beta}(b).
\end{equation*}
Recalling the definitions of $a,b,c$ and $d$, and \eqref{eq:sharma_m_as_tsa} we obtain
\[
S_{\alpha,\beta}(\bw)+S_{\alpha,\beta}(\bv)
\ge
S_{\alpha,\beta}(\bp)+S_{\alpha,\beta}(\bq),
\]
which demonstrates that  $S_{\alpha,\beta}$ is supermodular in the range $\alpha>1$ and $\beta\le \alpha$.

\smallskip
\noindent
\textbf{Case:} $\alpha=1$. We recall that the Sharma-Mittal entropy $S_{\alpha,\beta}(\bp)$ can be expressed as a function of the R\'enyi entropy $H_\alpha(\bp)$ as shown in \eqref{S=R}. Thus, as $\alpha\to1$, the R\'enyi entropy converges to the Shannon entropy $H(p)$, yielding the following representation:
\begin{equation}\label{eq:S_1}
    S_{1,\beta}(\bp) =\phi_\beta(H(\bp)),
\end{equation}
where the function $\phi_\beta(x)$ is defined as:
\begin{equation}\label{eq:phi}
    \phi_\beta(x)=\begin{cases}
\frac{2^{(1-\beta)x}-1}{1-\beta}, & \mbox{ for }\beta\neq 1,\\
\ln(2)x, &\mbox{ for } \beta=1.
\end{cases}
\end{equation}
Given that the Shannon entropy $H(\bp)$ is Schur-concave, subadditive and supermodular on the majorization lattice \cite{cv}, and that, for $\beta\leq\alpha=1$, the first derivative $\phi'_\beta(x)$ and second derivative $\phi''_\beta$ of the function $\phi_\beta(x)$ are non negative, we can proceed following the same idea of the Case $\alpha\in(0,1)$.
Let
\[
a=H(\bp),
\qquad
b=H(\bq),
\qquad
c=H(\bw),
\qquad
d=H(\bv).
\]
Because Shannon entropy $H(\bp)$ is Schur-concave and
\[
\bw\preceq \bp,\; \bq\preceq \bv,
\]
it follows
\[
c\ge a,
\qquad
c\ge b,
\qquad
a\ge d,
\qquad
b\ge d.
\]
Moreover, without loss of generality, assuming $a\geq b$, yields
\begin{equation}\label{eq:1}
c\ge a\ge b\ge d.
\end{equation}
Since the Shannon entropy satisfies the supermodularity property, we have
\begin{equation}\label{eq:2}
c+d\ge a+b.
\end{equation}
Let us define 
\[
\delta=(c-a)\ge0.
\]
From (\ref{eq:2}), it holds that
\[
d-b\ge -\delta.
\]
Hence there exists $\varepsilon\ge0$ such that
\begin{equation}\label{eq:3}
d=b-\delta+\varepsilon.
\end{equation}
Therefore
\[
(c,d)
=
(a+\delta,\ b-\delta+\varepsilon).
\]
Define the function $F(t)$ as
\[
F(t)=\phi_\beta(a+t)+\phi_\beta(b-t),
\qquad
0\le t\le \delta.
\]
where $\phi_\beta$ is defined in (\ref{eq:phi}).
We first check that the argument $b-t\geq b-\delta$ remains non negative, therefore it
stays within the domain of $\phi_\beta.$
Since the Shannon entropy is subadditive on the majorization lattice, we get that $c=H(\bp\wedge \bq)\leq H(\bp)+H(\bq)=a+b.$
Therefore $c\leq a+b$ ensures $b-t\geq b-\delta=a+b-c\geq 0$ for $t\in[0,\delta].$

Because $\beta\le\alpha=1$, we have that  $\phi_\beta''(x)=(\ln 2)^2(1-\beta)2^{(1-\beta)x}\ge0$. Therefore, 
the function $\phi_\beta$ is convex, and its derivative $\phi_\beta'$ is increasing.
Differentiating $F(t)$ yields
\[
F'(t)
=
\phi_\beta'(a+t)-\phi_\beta'(b-t).
\]
From the chain of inequalities in (\ref{eq:1}), it follows that
\[
a+t\ge b-t,
\]
and since $\phi_\beta'$ is increasing, we obtain that
\[
F'(t)\ge0.
\]
Thus, the function $F(t)$ is increasing in $t$.
Evaluating $F(t)$ at $t=\delta$ and $t=0$ gives
\begin{equation}\label{eq:4}
F(\delta)=\phi_\beta(a+\delta)+\phi_\beta(b-\delta)
\ge
\phi_\beta(a)+\phi_\beta(b)=F(0).
\end{equation}
Finally, since $\phi_\beta$ is increasing and $\varepsilon\ge0$, we obtain the inequality
\begin{equation}\label{ineq}
\phi_\beta(b-\delta+\varepsilon)
\ge
\phi_\beta(b-\delta).
\end{equation}
Combining (\ref{ineq}) with (\ref{eq:4}),
we get
\[
\phi_\beta(c)+\phi_\beta(d)
=
\phi_\beta(a+\delta)+\phi_\beta(b-\delta+\varepsilon)
\ge
\phi_\beta(a)+\phi_\beta(b).
\]
Recalling the definitions of $a,b,c, d$, and (\ref{eq:S_1}) we obtain
\[
S_{1,\beta}(\bw)+S_{1,\beta}(\bv)
\ge
S_{1,\beta}(\bp)+S_{1,\beta}(\bq),
\]
which demonstrates that  $S_{\alpha,\beta}$ is also supermodular for $\alpha=1$ and $\beta\le \alpha$.

\end{proof}

\remove{
We notice that the authors of \cite{ba} claim that the Tsallis entropy is submodular on the majorization lattice. 
However, their claim is incorrect. Our Theorem \ref{thm:sper} and
the fact that $T_\alpha(\mathbf{p})=\lim_{\beta\to \alpha}S_{\alpha,\beta}(\bp)$ allows us to state the
correct claim that the  Tsallis entropy is, indeed, supermodular on the majorization lattice for $\alpha>0$. Through Theorem \ref{thm:sper}, we also recover the results of \cite{yadav2025}. In fact, as $\beta\to1$, the Sharma-Mittal entropy reduces to the R\'enyi entropy of order $\alpha$, showing the the R\'enyi entropy is supermodular on the majorization lattice for $\alpha\geq 1$.
Finally, as a direct consequence of Corollary \ref{cor:super}, the Sharma-Mittal entropy is supermodular on the majorization lattice also  for $\alpha<0$ and $\beta\leq 1$.
}

\section{The behavior of $S_{\alpha,\beta}$ for  $\beta>\alpha$}
In this section, we show  that the Sharma-Mittal entropy $S_{\alpha,\beta}(\bp)$ is generally \textit{neither supermodular nor submodular} on the majorization lattice for $\beta > \alpha$, and $\alpha>0$. To this purpose,
we  construct two explicit counterexamples by finding one pair of probability distributions that breaks supermodularity and another pair that breaks submodularity.
Recall the definition of the Sharma-Mittal entropy:
\begin{equation}
    S_{\alpha, \beta}(\mathbf{p}) = \frac{1}{1-\beta} \left[ \left( \sum_{i=1}^n p_i^\alpha \right)^{\frac{1-\beta}{1-\alpha}} - 1 \right].
\end{equation}

For our counterexamples, let us fix $\alpha = 2$ and $\beta = 3$. This choice  simplifies the entropy formula to:
\begin{equation}
    S_{2,3}(\mathbf{p}) = \frac{1}{-2} \left[ \left( \sum_{i=1}^n p_i^2 \right)^{\frac{-2}{-1}} - 1 \right] = \frac{1}{2} \left( 1 - \left( \sum_{i=1}^n p_i^2 \right)^2 \right)
\end{equation}

We also recall that in the majorization lattice, given two probability distributions $\mathbf{p}$ and $\mathbf{q}$ ordered in a non-increasing fashion, their least upper bound $\mathbf{p} \vee \mathbf{q}$ and greatest lower bound $\mathbf{p} \wedge \mathbf{q}$ are constructed by taking the unique distributions whose cumulative sums correspond, respectively, to the component-wise maximum and minimum of the cumulative sums of $\mathbf{p}$ and $\mathbf{q}$ (see \eqref{eq:lub} and \eqref{eq:glb} for further details).

We can now present the two counterexamples in the following subsections.

\subsection{Counterexample 1: $S_{\alpha, \beta}(\mathbf{p})$ is not supermodular for $\beta>\alpha$}

In order to disprove supermodularity, we must show an instance, i.e., a pair of probability distributions $\mathbf{p}$ and $\mathbf{q}$, where the property fails, that is, $$S_{2,3}(\mathbf{p} \vee \mathbf{q}) + S_{2,3}(\mathbf{p} \wedge \mathbf{q}) < S_{2,3}(\mathbf{p}) + S_{2,3}(\mathbf{q}).$$

For this purpose, let $n=4$ and consider the following distributions:
\begin{equation*}
    \mathbf{p} = (0.5, 0.3, 0.1, 0.1)\quad\text{and }\quad\mathbf{q} = (0.4, 0.4, 0.2, 0.0).
\end{equation*}
One can verify that in the majorization lattice their least upper bound and greatest lower bound are, respectively:
\begin{equation*}
    \mathbf{p} \vee \mathbf{q} = (0.5, 0.3, 0.2, 0.0)\quad\text{and }\quad\mathbf{p} \wedge \mathbf{q} = (0.4, 0.4, 0.1, 0.1).
\end{equation*}
Let us evaluate their Sharma-Mittal entropy $S_{\alpha,\beta}$ for $\alpha=2$ and $\beta=3>1$:
\begin{align*}
    S_{2,3}(\mathbf{p})&=\frac{1}{2}-\frac{1}{2}\left(0.5^2+0.3^2+0.1^2+0.1^2\right)^2={0.4352}\\
    S_{2,3}(\mathbf{q})&=\frac{1}{2}-\frac{1}{2}\left(0.4^2+0.4^2+0.2^2+0.0^2\right)^2={0.4352}\\
    S_{2,3}(\mathbf{p} \vee \mathbf{q})&=\frac{1}{2}-\frac{1}{2}\left(0.5^2+0.3^2+0.2^2+0.0^2\right)^2={0.4278}\\
    S_{2,3}(\mathbf{p} \wedge \mathbf{q})&=\frac{1}{2}-\frac{1}{2}\left(0.4^2+0.4^2+0.1^2+0.1^2\right)^2={0.4422}
\end{align*}
Then, one can see that for the chosen pair of probability distributions, the following inequality holds
\begin{equation*}
    S_{2,3}(\mathbf{p})+S_{2,3}(\mathbf{q})={0.8704}>{0.8700}=S_{2,3}(\mathbf{p} \vee \mathbf{q})+S_{2,3}(\mathbf{p} \wedge \mathbf{q}),
\end{equation*}
which shows that $S_{2,3}$ is \textit{not} supermodular.

\subsection{Counterexample 2: $S_{\alpha, \beta}(\mathbf{p})$ is not submodular for $\beta>\alpha$}

In a similar fashion to the previous section, to disprove submodularity, we must exhibit an instance where the submodular property fails, that is, $$S_{2,3}(\mathbf{p} \vee \mathbf{q}) + S_{2,3}(\mathbf{p} \wedge \mathbf{q}) > S_{2,3}(\mathbf{p}) + S_{2,3}(\mathbf{q}).$$

Let $n=4$ and consider the following pair of distributions:
\begin{equation*}
    \textbf{p}=(0.5,0.2,0.2,0.1)\quad\text{and }\quad\textbf{q}=(0.4,0.4,0.15,0.05).
\end{equation*}
Their respective least upper bound and greatest lower bound are:
\begin{equation*}
    \mathbf{p} \vee \mathbf{q}=(0.5,0.3,0.15,0.05)\quad\text{and }\quad\mathbf{p} \wedge \mathbf{q}=(0.4,0.3,0.2,0.1).
\end{equation*}
Evaluating their Sharma-Mittal entropy $S_{\alpha,\beta}$ for $\alpha=2$ and $\beta=3$ yields:
\begin{align*}
    S_{2,3}(\mathbf{p})&=\frac{1}{2}-\frac{1}{2}\left(0.5^2+0.2^2+0.2^2+0.1^2\right)^2={0.4422}\\
    S_{2,3}(\mathbf{q})&=\frac{1}{2}-\frac{1}{2}\left(0.4^2+0.4^2+0.15^2+0.05^2\right)^2={0.4404875}\\
    S_{2,3}(\mathbf{p} \vee \mathbf{q})&=\frac{1}{2}-\frac{1}{2}\left(0.5^2+0.3^2+0.15^2+0.05^2\right)^2={0.4333875}\\
    S_{2,3}(\mathbf{p} \wedge \mathbf{q})&=\frac{1}{2}-\frac{1}{2}\left(0.4^2+0.3^2+0.2^2+0.1^2\right)^2={0.455}
\end{align*}
Consequently, for this choice of probability distributions, the following inequality holds
\begin{equation*}
    S_{2,3}(\mathbf{p} \vee \mathbf{q})+S_{2,3}(\mathbf{p} \wedge \mathbf{q})=0.8883875>0.8826875= S_{2,3}(\mathbf{p})+ S_{2,3}(\mathbf{q}),
\end{equation*}
which demonstrates that $S_{2,3}$ is \textit{not} submodular.

\medskip
Taken together, these two counterexamples demonstrate that for $\beta > \alpha$, and $\alpha>0$, the Sharma-Mittal entropy is \textit{neither supermodular nor submodular} on the majorization lattice.

\end{document}